\documentclass{article}

\usepackage[english]{babel}
\usepackage[utf8]{inputenc}
\usepackage{amsmath}
\usepackage{amssymb}
\usepackage{amsfonts} 
\usepackage{amssymb}
\usepackage{amsthm}
\usepackage[numbers]{natbib}
\usepackage{comment}
\usepackage{float}
\usepackage{color}
\usepackage{booktabs}
\usepackage{graphicx}
\usepackage[numbers]{natbib}
\usepackage{algorithm}
\usepackage{algpseudocode}
\usepackage{authblk}

\usepackage[a4paper]{geometry}
\newtheorem{theorem}{Theorem}

\newtheorem{lemma}{Lemma}
\newtheorem{proposition}{Proposition}

\begin{document}

\title{A 3-Approximation Algorithm for a Particular Case of the Hamiltonian $p$-Median Problem} 
\author{Michel Wan Der Maas Soares \\ {\tt michewandermaas@gmail.com} \and Dilson Lucas Pereira \\ {\tt dilson.pereira@ufla.br} \\ Departamento de Computação Aplicada \\ Universidade Federal de Lavras}
\date{2019}

\maketitle

\section*{Abstract} 
Given a weighted graph $G$ with $n$ vertices and $m$ edges, and a positive integer $p$, the Hamiltonian $p$-median problem consists in finding $p$ cycles of minimum total weight such that each vertex of $G$ is in exactly one cycle. We introduce an $O(n^6)$ 3-approximation algorithm for the particular case in which $p \leq \lceil \frac{n-2\lceil \frac{n}{5} \rceil}{3} \rceil$. An approximation ratio of 2 might be obtained depending on the number of components in the optimal 2-factor of $G$. We present computational experiments comparing the approximation algorithm to an exact algorithm from the literature. In practice much better ratios are obtained. For large values of $p$, the exact algorithm is outperformed by our approximation algorithm.

\noindent {\bf Keywords:} Hamiltonian $p$-median problem, approximation algorithms, 2-factor, matching

\section{Introduction}

Given a graph $G=(V,E)$ with $n$ vertices and $m$ edges, a positive integer $p$, and a set of edge weights $\{c_e \in \mathbb{R}_+: e \in E\}$, the Hamiltonian $p$-median problem (H$p$MP) asks for $p$ cycles of minimum total weight, such that each vertex of $G$ is in exactly one cycle. A cycle is defined as a sequence of vertices $(v_1, v_2, \ldots, v_k, v_1)$ such that $k \geq 3$, $\{v_k, v_1\} \in E$, $\{v_i, v_{i+1}\} \in E$ for $1 \leq i < k$, and $v_i \neq v_j$, for $1 \leq i < j \leq k$. Note that the problem has a feasible solution only if $p \leq \left\lfloor \frac{n}{3} \right\rfloor$. In what follows, we assume that $E$ is complete and the edge weights satisfy the triangle inequality. 

A particular case of the H$p$MP, when $p=1$, is the traveling salesman problem (TSP) \cite{DanFulJoh54}, one of the best known NP-Hard combinatorial optimization problems. Consequently, the H$p$MP is also NP-Hard. In fact, the TSP can be reduced to the H$p$MP for any value of $p$. 

The H$p$MP was introduced by Branco and Coelho \cite{BraCoe90}.
Formulations and exact approaches were studied in \cite{Glaab00, Zohrehbandian07, HpMP-INOCold, GolGouLap14, MarCenUst16, ErdLapChi18}. A polyhedral study was conducted in \cite{HupLie13}. An iterated local search heuristic was introduced in \cite{ErdLapChi18}. Not many heuristic approaches have been proposed for the problem.

Since polynomial time approximation algorithms exist for the TSP \cite{Vaz01}, a natural research question is whether such algorithms exist for the H$p$MP. In this paper, we introduce a 3-approximation algorithm that runs in $O(n^6)$ for the particular case of the H$p$MP in which $p \leq \lceil \frac{n-2\lceil \frac{n}{5} \rceil}{3} \rceil$. To the best of our knowledge, this is the first approximation algorithm proposed for the problem.

The rest of the paper is organized as follows. The approximation algorithm is introduced in Section \ref{sec:alg}, along with complexity and approximation results. In Section \ref{sec:exp}, we present some experimental results, comparing the approximation to an exact algorithm. Concluding remarks are given in Section \ref{sec:conc}.

\section{The Approximation Algorithm}\label{sec:alg}
Before presenting the algorithm, we introduce some notation. The notation $S \subseteq G$ is used to indicate that $S$ is a subgraph of $G$. Given $S \subseteq G$, we let $V(S)$ denote its set of vertices and $E(S)$ denote its set of edges. We define the weight of $S$ as $c(S) = \sum_{e \in E(S)} c_e$.  We refer to cycles containing exactly $k$ vertices as $k$-cycles. The optimal H$p$MP solution is denoted $H^*$.

The approximation algorithm we propose is based on the concept of 2-factors. A 2-factor of a graph is a set of cycles such that each vertex is in exactly one cycle. A related concept is that of 2-matchings, the difference being that 2-matchings might have cycles containing only 2 vertices. Given a weighted graph, a 2-factor having minimum total weight can be found in polynomial time, using matching techniques. One possible approach, presented in \cite{CooCunPul97}, is reproduced in Section \ref{sec:2f}.
As any feasible H$p$MP solution is a 2-factor, the minimum weight 2-factor yields a lower bound on the optimal H$p$MP solution.

The approximation algorithm is presented in Algorithm \ref{alg}.

\begin{algorithm}[htpb]
\begin{algorithmic}[1]
\State Find a minimum weight 2-factor $F$ of $G$. Note that $c(F) \leq c(H^*)$. Let $q$ be the number of cycles in $F$. 
\If {$q = p$} 
    \State $H \leftarrow F$.
\Else
    \If {$q > p$} 
        \State Find a spanning tree $T$ of $G$, remove its $p-1$ edges with largest weight. Note that after the removal $c(T) \leq c(H^*)$ and the number of components in $T$ is $p$, but some of its components  might be singletons.
        \State Since $T$ has $p$ components, its union with $F$ will create a graph with $p$ or less components. Add to $F$ edges of $T$ connecting different components of $F$ until it has $p$ components. Note that after this $c(F) \leq 2c(H^*)$.
        \State Transform $F$ into a multigraph by duplicating the edges that were added from $T$, $c(F) \leq 3c(H^*)$.
    \ElsIf {$q < p$}
        \While {$F$ does not have $p$ components}
            \State Select a component of $F$ with at least six vertices. For the special case of $p \leq \lceil \frac{n-2\lceil \frac{n}{5} \rceil}{3} \rceil$ such a component always exists. This component is either a cycle $(v_1, v_2, \ldots, v_k, v_1)$ or a path $(v_1, v_2, \ldots, v_k)$. Irrespective of the case, replace it by two paths $(v_1, v_2, v_3)$ and $(v_4, \ldots, v_k)$. Note that this does not increase $c(F)$, i.e., $c(F) \leq c(H^*)$. 
        \EndWhile
        \State Transform $F$ into a multigraph by duplicating all of its edges, $c(F) \leq 2c(H^*)$.
    \EndIf
    \State Note that $F$ is a multigraph with $p$ connected components in which the degree of each vertex is even. The approximate solution $H$ is obtained by following an Eulerian tour for each component of $F$, skipping vertices already visited.
\EndIf
\State {\bf Return} $H$. If case $q=p$ was executed $H$ is an optimal solution. If case $q>p$ was executed, $H$ is a 3-approximation. If case $q<p$ was executed, $H$ is a 2-approximation.
\end{algorithmic}
\caption{H$p$MP approximation algorithm.}
\label{alg}
\end{algorithm}

For general values of $p$, it might happen that in step 11 the number of components in $F$ is less than $p$ and there are no components with six or more vertices. As an example, if $n=10$, $p=3$, and the optimal 2-factor has 2 connected components with 5 vertices each, the algorithm will fail. This type of problem is caused by components of $F$ whose number of vertices is not a multiple of 3. When splitting components in step 11, a component with $k$ vertices will have $k \mod 3$ vertices that cannot be left alone, as they are not enough to form a cycle. If we try to connect them to vertices coming from other components we might violate the approximation ratio guarantee. For the special case we consider, $p \leq \lceil \frac{n-2\lceil \frac{n}{5} \rceil}{3} \rceil$, a component with at least six vertices will always exist in step 11. This is proven below, we will need the following result first:

\begin{lemma}
Given a 2-factor $F$ of $G$ composed of $q$ cycles $\{C_1, \ldots, C_q\}$ let $l(F) = \sum_{i = 1}^{q} (|V(C_i)| \mod 3)$ be the number of vertices that cannot form cycles by themselves in step 11 of Algorithm \ref{alg}. Then, the maximum possible value of $l(F)$ over all 2-factors of $G$ is not greater than $2\left\lceil \frac{n}{5} \right\rceil$.
\end{lemma}
\begin{proof}
We first show that the maximum of $l(F)$ is attained when $F$ is a 2-factor having the maximum possible number of 5-cycles. Let $F^*$ be a 2-factor attaining the maximum value of $l$. Let $r$ be the size of the largest cycle of $F^*$. If $F^*$ is not a 2-factor with the maximum possible number of 5-cycles, we can always construct another two factor from $F^*$ having one more 5-cycle than $F^*$ without decreasing $l(F^*)$.  
If $r \geq 8$, replace the corresponding cycle by two cycles, one containing 5 and the other containing the remaining vertices. If $r \leq 7$, there must be another cycle with number of vertices $r'$ such that $r'\leq 7$ and $r' \neq 5$, otherwise $F^*$ has the maximum possible number of 5-cycles. Now, it is a matter of examining the possibilities and showing that for each of them we can construct another 2-factor having more 5-cycles than $F^*$ yielding the same value (or a larger value) of $l(F^*)$. The possibilities are:
\begin{itemize}
\item $r = 7$ and $r' = 7$, we can generate cycles with 5, 5, and 4 vertices, respectively.
\item $r = 7$ and $r' = 6$, we can generate cycles with 5, 5, and 3 vertices, respectively.
\item $r = 7$ and $r' = 4$, we can generate cycles with 5 and 6 vertices, respectively.
\item $r = 7$ and $r' = 3$, we can generate cycles with 5 and 5 vertices, respectively.
\item $r = 6$ and $r' = 6$, we can generate cycles with 5 and 7 vertices, respectively.
\item $r = 6$ and $r' = 4$, we can generate cycles with 5 and 5 vertices, respectively.
\item $r = 6$ and $r' = 3$, we can generate cycles with 5 and 4 vertices, respectively.
\item $r = 4$ and $r' = 4$, we can generate cycles with 5 and 3 vertices, respectively.
\item $r = 4$ and $r' = 3$, $F^*$ already has the maximum number of 5-cycles.

\end{itemize}
Now, if $F$ has the maximum number of 5-cycles, than it has at least $\left\lfloor \frac{n}{5} \right\rfloor-1$ 5-cycles. The possible values for the remaining $n-(\left\lfloor \frac{n}{5} \right\rfloor-1)$ vertices are $5, \ldots, 9$. If there are 5 remaining vertices, the maximum of $l(F)$ is attained by organizing them into a 5 cycle, in which case $l(F) = 2\frac{n}{5}$. If there are 6 remaining vertices, the maximum of $l(F)$ is attained by organizing them into a 6-cycle, in which case $l(F) = 2(\left\lfloor \frac{n}{5} \right\rfloor - 1)$. If there are 7, the maximum is attained by organizing them into a 7-cycle, $l(F) = 2(\left\lfloor \frac{n}{5} \right\rfloor -1) + 1$. If there are 8, by organizing them into a 5-cycle and a 3-cycle, $l(F) = 2\left\lfloor \frac{n}{5} \right\rfloor$. And if there are 9, by organizing them into a 5-cycle and a 4-cycle, $l(F) = 2\left\lfloor \frac{n}{5} \right\rfloor + 1$. In all these cases, $l(F) \leq 2\left\lceil \frac{n}{5} \right\rceil $.
\end{proof}

\begin{proposition}
\label{pro}
It is possible to split any 2-factor of a graph $G$ into at least $\left\lceil \frac{n - 2\left\lceil \frac{n}{5} \right\rceil }{3} \right\rceil$ connected components containing at least 3 vertices each.
\end{proposition}
\begin{proof}
Given a 2-factor $F$ of $G$ composed of $q$ cycles $\{C_1, \ldots, C_q\}$, the maximum number of connected components with at least 3 vertices it can be split into is $\frac{n-l(F)}{3} \geq \left\lceil \frac{n - 2\left\lceil \frac{n}{5} \right\rceil }{3} \right\rceil$, since $l(F) \leq 2\left\lceil \frac{n}{5} \right\rceil$ and the number of resulting components is integer.
\end{proof}

Thus, for $p \leq \left\lceil \frac{n - 2\left\lceil \frac{n}{5} \right\rceil }{3} \right\rceil$, Algorithm \ref{alg} is guaranteed to find a feasible solution.

\begin{theorem}
For the special case of the H$p$MP with $p \leq \lceil \frac{n-2\lceil \frac{n}{5} \rceil}{3} \rceil$, Algorithm \ref{alg} finds a solution $H$ with $c(H) \leq 3c(H^*)$ in $O(n^6)$ time.
\end{theorem}
\begin{proof}
The dominating term in the computational complexity of Algorithm \ref{alg} is given by step 1. The minimum weight 2-factor of a graph $G$ can be found by executing a minimum weight perfect matching (MWPM) algorithm on an auxiliary graph $G'$ with $O(n^2)$ vertices and $O(m)$ edges, obtained from $G$. The MWPM of a graph $G$ can be found in $O(n^2m)$ \cite{Edm65}, thus, finding the MWPM of $G'$ is $O(n^6)$ as we assume $E$ to be complete. The approximation factor was proven along the algorithm description.
\end{proof}

For the sake of completeness, in the next section we present a scheme for finding a minimum weight 2-factor of $G$.

\subsection{Finding minimum weight 2-factors of $G$}\label{sec:2f}
Minimum weight 2-factors of $G$ can be found on a appropriately defined auxiliary graph, as described in \cite{CooCunPul97}, pp. 185.
First, define a graph $G' = \{V', E'\}$: Replace each edge $e =\{u,v\}$ of $E$ by two new vertices $a_e^u$ and $a_e^v$ and three new edges $\{u, a_e^u\}$, $\{a_e^u, a_e^v\}$, and $\{a_e^v, v\}$. 

Now, define a graph $G'' = \{V'', E''\}$ from $G'$: For each original vertex $v \in V \cap V'$, replace it by two new vertices $b'_v$ and $b_v''$. Replace every edge $\{v, a_e^v\}$ by two new edges $\{b'_v, a_e^v\}$ and $\{b''_v, a_e^v\}$, both with weight $c''_{\{b'_v, a_e^v\}} = c''_{\{b''_v, a_e^v\}} = \frac{c_e}{2}$. Let $c''_e = 0$ for any other edge in $E''$.

An edge $e = \{u,v\} \in E$ is in the minimum weight 2-factor of $G$ if and only if $\{b'_u, a_e^u\}$ or $\{b''_u, a_e^u\}$ is in the minimum weight perfect matching of $G''$, in which case either $\{b'_v, a_e^v\}$ or $\{b''_v, a_e^v\}$ will also be in the matching.

\section{Computational Experiments}\label{sec:exp}
In this section, we present results from computational experiments that were conducted to evaluate the practical approximation ratio of the algorithm introduced here, as well as its performance in comparison to an exact algorithm.

We consider the instances with 100 vertices among those proposed in Gollowitzer et al. \cite{HpMP-INOCold}, they comprise 5 complete Euclidean graphs whose vertices were randomly generated by sampling 100 points in the square $[0,100]^2$. These instances are numbered 1-5. Besides these 5 instances, 5 more instances were generated following the same scheme. These instances are numbered 6-10. Since $n=100$ in all instances, the special H$p$MP case for which our algorithm was devised consists of the instances with $p \leq 20$. Besides values of $p$ in this range we also considered larger values, in order to evaluate the algorithm's behavior for more general values of $p$. The following values of $p$ were considered: $\{2, 10, 18, 26, 33\}$.

Gollowitzer et al. \cite{GolGouLap14} propose several models from the H$p$MP. The exact algorithm we use in the experiments is the one based on ``model 1'', the best performing algorithm in that reference. We had access to its source code and ported it to the version 12.63 of the IBM CPLEX solver.

The computational experiments were carried in an Intel XEON machine with 8GB of RAM and 8 cores running at 3.5GHz, running under Linux Ubuntu 14.04. Both the exact and the approximation algorithm were implemented in C++ and evaluated on this same computational environment. A time limit of 1 hour was imposed for the exact algorithm. The source code for the approximation algorithm can be found in \cite{GITHUB}. 

The computational results are presented in Tables \ref{tab1} and \ref{tab2}. Table \ref{tab1} presents results for $p \in \{2, 10, 18\}$ while Table \ref{tab2} presents results for $p \in \{26, 33\}$. Under the heading ``Instance'' we present instance information: the value of $p$ and a number identifying the instance. Under the heading ``Exact Algorithm'' we present results for the exact algorithm: the lower ($lb$) and upper bound ($ub$) upon the algorithm termination, and the total computational time in seconds ($t(s)$). Under the heading ``Approximation Algorithm``, we present results for the approximation algorithm: the obtained upper bound ($ub$), the total computational time ($t(s)$), the approximation factor (apx) in relation to the lower bound provided the exact algorithm and the number of cycles ($q$) in the 2-factor computed in step 1 of Algorithm \ref{alg}.

\begin{table}[htbp]
\centering
\footnotesize
\begin{tabular}{@{}lrrrrrrrrrr@{}} 
\toprule
\multicolumn{2}{c}{Instance} & & \multicolumn{3}{c}{Exact Algorithm} & & \multicolumn{4}{c}{Approximation Algorithm}\\
\cmidrule{1-2}\cmidrule{4-6}\cmidrule{8-11}
$p$ & ID & & $lb$ & $ub$ & $t(s)$ & & $ub$ & $t(s)$ & apx & $q$\\
\midrule 
2 & 1 &  & 782.29 & 782.29 & 665.7 &  & 1018.72 & 2.57 & 1.3 & 16\\
& 2 &  & 781.52 & 781.52 & 38.7 &  & 920.02 & 2.99 & 1.17 & 11\\
& 3 &  & 732.24 & 732.24 & 132.95 &  & 850.01 & 2.49 & 1.16 & 9\\
& 4 &  & 789.62 & 789.62 & 559.4 &  & 873.28 & 2.87 & 1.1 & 9\\
& 5 &  & 748.19 & 748.19 & 134.2 &  & 962.85 & 2.62 & 1.28 & 16\\
& 6 &  & 761.12 & 761.12 & 73.12 &  & 815.12 & 2.63 & 1.07 & 7\\
& 7 &  & 754.79 & 754.79 & 76.56 &  & 922.95 & 2.57 & 1.22 & 14\\
& 8 &  & 752.55 & 783.93 & 3600 &  & 1067.03 & 2.62 & 1.41 & 22\\
& 9 &  & 778.12 & 778.12 & 235.31 &  & 1059.02 & 2.61 & 1.36 & 19\\
& 10 &  & 787.75 & 787.75 & 19.31 &  & 965.22 & 2.45 & 1.22 & 14\\
\cmidrule{2-11}
& avg. &  & 766.81 & 769.95 & 553.52 &  & 945.42 & 2.64 & 1.22 & 13.7\\
\midrule 
10 & 1 &  & 758.52 & 758.52 & 5.9 &  & 838.53 & 2.54 & 1.1 & 16\\
& 2 &  & 770.25 & 770.25 & 3.9 &  & 780.97 & 2.93 & 1.01 & 11\\
& 3 &  & 714.49 & 714.49 & 2.39 &  & 742.5 & 2.63 & 1.03 & 9\\
& 4 &  & 764.74 & 764.74 & 2.44 &  & 796.01 & 2.77 & 1.04 & 9\\
& 5 &  & 724.61 & 724.61 & 15.03 &  & 798.97 & 2.58 & 1.1 & 16\\
& 6 &  & 748.96 & 748.96 & 50.68 &  & 776.61 & 2.61 & 1.03 & 7\\
& 7 &  & 737.49 & 737.49 & 7.17 &  & 771.79 & 2.44 & 1.04 & 14\\
& 8 &  & 721.92 & 721.92 & 95.25 &  & 879.43 & 2.53 & 1.21 & 22\\
& 9 &  & 751.18 & 751.18 & 18.49 &  & 885.37 & 2.6 & 1.17 & 19\\
& 10 &  & 767.77 & 767.77 & 1.09 &  & 811.65 & 2.45 & 1.05 & 14\\
\cmidrule{2-11}
& avg. &  & 745.99 & 745.99 & 20.23 &  & 808.18 & 2.60 & 1.07 & 13.7\\
\midrule 
18 & 1 &  & 754.15 & 754.15 & 11.04 &  & 779.62 & 2.55 & 1.03 &16 \\
& 2 &  & 776.21 & 782.54 & 3600 &  & 843.03 & 2.98 & 1.08 & 11 \\
& 3 &  & 720.43 & 735.43 & 3600 &  & 764.55 & 2.64 & 1.06 & 9 \\
& 4 &  & 769.5 & 770.23 & 3600 &  & 883.84 & 2.69 & 1.14 & 9 \\
& 5 &  & 723.34 & 723.34 & 12.19 &  & 773.38 & 2.56 & 1.06 & 16 \\
& 6 &  & 750.09 & 750.09 & 32.62 &  & 871.85 & 2.63 & 1.16 & 7 \\
& 7 &  & 734.91 & 734.91 & 18.61 &  & 817.68 & 2.51 & 1.11 & 14 \\
& 8 &  & 709.13 & 709.13 & 1.63 &  & 740.2 & 2.53 & 1.04 & 22 \\
& 9 &  & 744.04 & 744.04 & 1.83 &  & 754.58 & 2.6 & 1.01 & 19 \\
& 10 &  & 772.45 & 772.72 & 3600 &  & 820.37 & 2.45 & 1.06 & 14 \\
\cmidrule{2-11}
& avg. &  & 745.42 & 747.65 & 1447.79 &  & 804.91 & 2.614 & 1.07 & 13.7\\
\bottomrule
\end{tabular}
\caption{Results for instances with $n=100$ and $p \in \{2, 10, 18\}$.}
\label{tab1}
\end{table}

\begin{table}[htbp]
\centering
\footnotesize
\begin{tabular}{@{}lrrrrrrrrrr@{}} 
\toprule
\multicolumn{2}{c}{Instance} & & \multicolumn{3}{c}{Exact Algorithm} & & \multicolumn{4}{c}{Approximation Algorithm}\\
\cmidrule{1-2}\cmidrule{4-6}\cmidrule{8-11}
$p$ & ID & & $lb$ & $ub$ & $t(s)$ & & $ub$ & $t(s)$ & apx & $q$\\
\midrule 
26 & 1 &  & 763.22 & 779.39 & 3600 &  & 821.87 & 2.56 & 1.07 &16 \\
& 2 &  & 780.01 & 838.97 & 3600 &  & 882.28 & 3.1 & 1.13 & 11 \\
& 3 &  & 724.36 & 776.42 & 3600 &  & 852.75 & 2.48 & 1.17 & 9 \\
& 4 &  & 773.17 & 856.18 & 3600 &  & 885.26 & 2.75 & 1.14 & 9 \\
& 5 &  & 730.45 & 752.31 & 3600 &  & 825.69 & 2.56 & 1.13 & 16 \\
& 6 &  & 757.35 & 814.33 & 3600 &  & 917.27 & 2.63 & 1.21 & 7 \\
& 7 &  & 743.99 & 768.43 & 3600 &  & 853.86 & 2.51 & 1.14 & 14 \\
& 8 &  & 713.33 & 714.26 & 3600 &  & 750.86 & 2.52 & 1.05 & 22 \\
& 9 &  & 750.11 & 757.43 & 3600 &  & 784.2 & 2.59 & 1.04 & 19 \\
& 10 &  & 777.39 & 795.51 & 3600 &  & 881.95 & 2.44 & 1.13 & 14 \\
\cmidrule{2-11}
& avg. &  & 751.33 & 785.32 & 3600 &  & 845.59 & 2.61 & 1.12 & 13.7\\
\midrule 
33 & 1 &  & 769.34 & 1006.93 & 3600 &  & 826.08 & 2.57 & 1.07 &16 \\
& 2 &  & 783.96 & 1086.46 & 3600 &  & 877.14 & 3.1 & 1.11 & 11 \\
& 3 &  & 729.49 & 962.52 & 3600 &  & 838.51 & 2.59 & 1.14 & 9 \\
& 4 &  & 778.92 & 1092.88 & 3600 &  & 869.3 & 2.74 & 1.11 & 9 \\
& 5 &  & 736.38 & 893.25 & 3600 &  & 825.47 & 2.56 & 1.12 & 16 \\
& 6 &  & 762.55 & 1195.72 & 3600 &  & 930.87 & 2.75 & 1.22 & 7 \\
& 7 &  & 749.7 & 1044.82 & 3600 &  & 858.56 & 2.66 & 1.14 & 14 \\
& 8 &  & 717.85 & 1045.19 & 3600 &  & 766.35 & 2.5 & 1.06 & 22 \\
& 9 &  & 754.22 & 891.91 & 3600 &  & 833.5 & 2.62 & 1.1 & 19 \\
& 10 &  & 783.52 & 1149.15 & 3600 &  & 908.93 & 2.47 & 1.16 & 14 \\
\cmidrule{2-11}
&avg.  &  & 756.59 & 1036.88 & 3600 &  & 853.47 & 2.65 & 1.12 & 13.7\\
\bottomrule
\end{tabular}
\caption{Results for instances with $n=100$ and $p \in \{26, 33\}$.}
\label{tab2}
\end{table}

In terms of performance, the approximation algorithm took between 2.4 and 3.1 seconds to solve the instances, with no significant difference in time across the different values of $p$. The exact algorithm, on the other hand, was not able to find the optimal solution for large values of $p$ under the time limit of one hour, and for the smallest values the time taken varied significantly from instance to instance.

In terms of solution quality, the approximation algorithm achieved practical approximation ratios between 1.01 and 1.41. The exact algorithm was able to find the optimal solution for low values of $p$, and obtained relatively small optimality gaps for instances with mid-range values of $p$. However, under the time limit of one hour, for the $p$ value of 33, the approximation algorithm was able to find better solutions for all instances.

On average, the value of $q$ was 13.7, which explains the fact that a better practical approximation ratio was obtained for instances with $p=10$ and $p=18$. The quality of the approximation also seems to be better among the $p > q$ cases, in which the average gap was 1.10 as opposed to 1.16 for the other case.

These results highlight the efficiency of the the approximation algorithm: small computational times and a good practical approximation factor, which in practice was much better than the theoretical factor of 3.

\section{Conclusion}\label{sec:conc}
In this paper, we introduced a 3-approximation algorithm for the special $p \leq \lceil \frac{n-2\lceil \frac{n}{5} \rceil}{3} \rceil$ case of the H$p$MP. The first H$p$MP approximation algorithm in the literature to the best of our knowledge.

We conducted computation experiments to assess the practical approximation ratio of the proposed algorithm, comparing it to an exact algorithm from the literature. The practical approximation ratio was 1.1, on average, with computational times orders of magnitude smaller than the exact algorithm. In fact, the approximation algorithm proves to be specially helpful for higher values of $p$, for which the exact algorithm has a hard time finding the optimal solution. 

As future research, it might be interesting to look for new approximation algorithms, with better approximation ratios and time complexities, and capable of handling general values of $p$. Another possible line of research is the combination of the approximation algorithm with local search procedures. 

\bibliographystyle{plainnat}
\bibliography{references}

\begin{thebibliography}{13}
\providecommand{\natexlab}[1]{#1}
\providecommand{\url}[1]{\texttt{#1}}
\expandafter\ifx\csname urlstyle\endcsname\relax
  \providecommand{\doi}[1]{doi: #1}\else
  \providecommand{\doi}{doi: \begingroup \urlstyle{rm}\Url}\fi

\bibitem[Branco and Coelho(1990)]{BraCoe90}
I.~M. Branco and J.~D. Coelho.
\newblock The hamiltonian p-median problem.
\newblock \emph{European Journal of Operational Research}, 47\penalty0
  (1):\penalty0 86 -- 95, 1990.

\bibitem[Cook et~al.(1997)Cook, Cunningham, Pulleyblank, and
  Schrijver]{CooCunPul97}
W.~J. Cook, W.~H. Cunningham, W.~R. Pulleyblank, and A.~Schrijver.
\newblock \emph{Combinatorial Optimization}.
\newblock Wiley, 1997.

\bibitem[Dantzig et~al.(1954)Dantzig, Fulkerson, and Johnson]{DanFulJoh54}
G.~Dantzig, R.~Fulkerson, and S.~Johnson.
\newblock Solution of a large-scale traveling-salesman problem.
\newblock \emph{Journal of the Operations Research Society of America},
  2\penalty0 (4):\penalty0 393--410, 1954.

\bibitem[Edmonds(1965)]{Edm65}
Jack Edmonds.
\newblock Paths, trees and flowers.
\newblock \emph{Canadian Journal of Mathematics}, pages 449--467, 1965.

\bibitem[Erdoğan et~al.(2016)Erdoğan, Laporte, and Chía]{ErdLapChi18}
G.~Erdoğan, G.~Laporte, and A.~M.~R. Chía.
\newblock Exact and heuristic algorithms for the hamiltonian p-median problem.
\newblock \emph{European Journal of Operational Research}, 253\penalty0
  (2):\penalty0 280 -- 289, 2016.

\bibitem[Glaab and Pott(2000)]{Glaab00}
H.~Glaab and A.~Pott.
\newblock {The Hamiltonian $p$-median problem}.
\newblock \emph{The Electronic Journal of Combinatorics}, 7\penalty0
  (R42):\penalty0 2, 2000.

\bibitem[Gollowitzer et~al.(2014)Gollowitzer, Gouveia, Laporte, Pereira, and
  Wojciechowski]{GolGouLap14}
S.~Gollowitzer, L.~Gouveia, G.~Laporte, D.~L. Pereira, and A.~Wojciechowski.
\newblock A comparison of several models for the hamiltonian p-median problem.
\newblock \emph{Networks}, 63\penalty0 (4):\penalty0 350--363, 2014.

\bibitem[Gollowitzer et~al.(2011)Gollowitzer, Pereira, and
  Wojciechowski]{HpMP-INOCold}
Stefan Gollowitzer, Dilson~Lucas Pereira, and Adam Wojciechowski.
\newblock New models for and numerical tests of the hamiltonian p-median
  problem.
\newblock In \emph{Volume 6701 of Lecture Notes in Computer Science}, 2011.
\newblock Proceedings of the International Network Optimization Conference
  (INOC) 2011.

\bibitem[Hupp and Liers(2013)]{HupLie13}
L.~Hupp and F.~Liers.
\newblock A polyhedral study of the hamiltonian p-median problem.
\newblock \emph{Electronic Notes in Discrete Mathematics}, 41:\penalty0 213 --
  220, 2013.

\bibitem[Maas()]{GITHUB}
Michel Wan~Der Maas.
\newblock Implementation of the approximation algorithm.
\newblock URL
  \url{https://github.com/michelwandermaas/p\_hamiltonian\_heuristic}.

\bibitem[Marzouk et~al.(2016)Marzouk, Moreno-Centeno, and \"Uster]{MarCenUst16}
A.~M. Marzouk, E.~Moreno-Centeno, and H.~\"Uster.
\newblock A branch-and-price algorithm for solving the hamiltonian p-median
  problem.
\newblock \emph{Informs Journal on Computing}, 28\penalty0 (4), 2016.

\bibitem[Vazirani(2001)]{Vaz01}
Vijay~V. Vazirani.
\newblock \emph{Approximation Algorithms}.
\newblock Springer-Verlag, Berlin, Heidelberg, 2001.

\bibitem[Zohrehbandian(2007)]{Zohrehbandian07}
M.~Zohrehbandian.
\newblock A new formulation of the {Hamiltonian} $p$-median problem.
\newblock \emph{Applied Mathematical Sciences}, 1\penalty0 (8):\penalty0
  355--361, 2007.

\end{thebibliography}
\end{document}